\documentclass[reprint,superscriptaddress,amsmath,amssymb,aps,prl,
twocolumn,nofootinbib]{revtex4-2}

\usepackage{graphicx}
\usepackage{booktabs}
\usepackage{tikz}
\usetikzlibrary{arrows.meta}
\usepackage{quantikz}
\usepackage{xcolor}
\usepackage[colorlinks=true,allcolors=blue!60!black]{hyperref}
\usepackage{orcidlink}   

\newtheorem{proposition}{Proposition}

\newcommand{\nm}{n_{\mathrm m}}
\newcommand{\nw}{n_{\mathrm w}}
\providecommand{\ket}[1]{\lvert #1 \rangle}
\providecommand{\bra}[1]{\langle #1 \rvert}
\DeclareMathOperator{\rank}{rank}

\begin{document}
	
	\title{Memory Separations from a Collective-Spin Quantum Register}
	
	\author{Jamal Slim\orcidlink{0000-0002-9418-8459}}
	\email{jamal.slim@desy.de}
	\affiliation{Deutsches Elektronen-Synchrotron DESY, 22603 Hamburg, Germany}
	
	\date{\today}
	
	\begin{abstract}
		A spin one-half driven only by total-spin operators generates a renewal
		process in which one interevent time $n^\ast$ is forbidden while all others
		occur. No classical hidden Markov model reproduces it with fewer than $1+\sqrt{n^\ast}$ states, an unbounded separation whose only resource is
		one relative phase of $\pi$. With $K$ amplitudes the same interference yields
		$K(K{+}1)/2$ decay channels, and for positive weights and generic rates the classical cost is
		exactly $K(K{+}1)/2$ states. Both are realized by a depth-three circuit whose
		parameter count does not grow with the register.
	\end{abstract}
	
	\maketitle
	
	\section{Introduction}
	\label{sec:intro}
	A stochastic process can be generated by a memory that is updated and read
	out one step at a time. The classical version is a hidden Markov model, and
	the relevant cost is the number of hidden states. The quantum version
	replaces the hidden states by a register that is never measured directly,
	and encodes them as states that need not be mutually
	orthogonal~\cite{gu}. Since the number of configurations is then decoupled
	from the dimension of the register, the quantum memory can be smaller than
	any classical one for the same process~\cite{elliott,elliottgu,zonnios},
	in some families without bound~\cite{garner,rare,longrange}, and the
	reduction has been demonstrated
	experimentally~\cite{palsson,ghafari,wu}.
	
	We study this in a register with a specific and experimentally natural
	restriction. The memory is driven only by \emph{collective} operations,
	meaning operations that act identically on every memory qubit and therefore
	depend on the register only through its total spin. Such a register is what
	an atomic ensemble or a spin-squeezed sample provides, and its dimension is
	$\nm+1$ rather than $2^{\nm}$ because the state never leaves the
	permutation-symmetric subspace. The register is a sequential Born machine with a
	coherent ancilla in the sense of Ref.~\cite{schon}, equivalently a hidden
	quantum Markov model~\cite{glasser}, specialized here to a memory that
	admits only collective operations.
	
	Two results follow. Setting one relative phase to $\pi$ makes the
	interevent distribution of the generated process vanish exactly at a chosen
	time $n^\ast$, and the classical cost of that single forbidden time grows
	without bound in $n^\ast$. Using $K$ amplitudes instead of two, the same
	interference generates $K(K{+}1)/2$ decay channels, and when every weight is positive the
	classical cost is exactly $K(K{+}1)/2$ states for generic decay rates, a
	separation that is attained rather than inferred from lower bounds. Both are
	realized by a depth-three circuit whose parameter count is independent of
	the register size.
	
	\section{The register}
	\label{sec:model}
	One block of the model prepares a work register of $\nw$ qubits in
	$\ket{0}$, applies a unitary $U(\theta)$, measures the work register in the
	computational basis, emits the outcome $x$, and resets. Writing
	\begin{equation}
		A_x=\big(\bra{x}_{\mathrm w}\otimes I_{\mathrm m}\big)\,U(\theta)\,
		\big(\ket{0}_{\mathrm w}\otimes I_{\mathrm m}\big),
		\label{eq:kraus}
	\end{equation}
	which acts on the memory alone and satisfies $\sum_x A_x^\dagger A_x=I$, the
	law of a length-$B$ output string is
	$p(x_1\ldots x_B)=\lVert A_{x_B}\cdots A_{x_1}\ket{0}\rVert^2$.
	
	The memory is driven only by the total-spin operators
	$J_\alpha=\tfrac12\sum_q\sigma_\alpha^{(q)}$ and by couplings
	$Z_k\otimes J_z$ to the work register. The initial state
	$\ket{0}^{\otimes\nm}$ is permutation symmetric and every collective
	operation preserves that symmetry, so the memory remains in the spin-$j$
	irreducible representation with $j=\nm/2$, of dimension $\nm+1$, spanned by
	the Dicke states~\cite{dicke}. Nothing is truncated. Two consequences
	matter here. The memory carries $\nm+1$ amplitudes rather than $2^{\nm}$,
	and a collective rotation carries one angle however many qubits it turns, so
	the number of parameters does not depend on $\nm$. One block is shown in
	Fig.~\ref{fig:circuit}.
	
	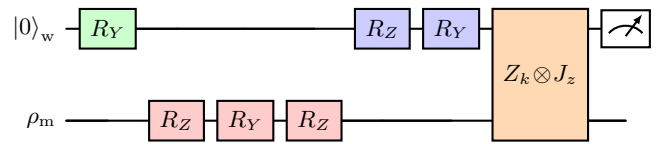
\begin{figure}[b]
		\centering
		\resizebox{\columnwidth}{!}{%
			\begin{quantikz}[column sep=0.18cm, row sep=0.7cm]
				\lstick{$\ket{0}_{\mathrm w}$}
				& \gate[style={fill=green!20}]{R_Y}
				& \qw & \qw & \qw
				& \gate[style={fill=blue!20}]{R_Z}
				& \gate[style={fill=blue!20}]{R_Y}
				& \gate[2,style={fill=orange!30}]{Z_k\!\otimes\! J_z}
				& \meter{} \\
				\lstick{$\rho_{\mathrm m}$}
				& \qw
				& \gate[style={fill=red!20}]{R_Z}
				& \gate[style={fill=red!20}]{R_Y}
				& \gate[style={fill=red!20}]{R_Z}
				& \qw & \qw &
				& \qw
			\end{quantikz}%
		}
		\caption{One block. The work register (top) is prepared blank, rotated
			(green, blue), coupled to the memory (orange) and measured, and the outcome
			is the emitted symbol. The memory (bottom) is never measured and its
			rotations (red) are collective, one angle each however many qubits they
			turn. The layer repeats $L$ times per block. The operators of
			Eq.~\eqref{eq:kraus} are the blocks $A_x=\bra{x}U\ket{0}$ of this circuit.}
		\label{fig:circuit}
	\end{figure}
	
	\section{A forbidden interevent time}
	\label{sec:forbidden}
	A renewal process emits events at random times, the intervals between
	successive events being independent and identically distributed. An observer
	samples at a fixed interval $\Delta t$, writing $1$ in a bin that contains an
	event and $0$ otherwise. Let $\Phi(n)$ be the probability that a run of
	zeros reaches length $n$, and $d_n=\Phi(n)-\Phi(n{+}1)$ the probability that
	it ends at $n$, so that $d_n$ is the interevent distribution.
	
	Suppose the memory carries two non-orthogonal generator states
	$\ket{\phi_1},\ket{\phi_2}$ with decay factors
	$\Gamma_j=e^{-\gamma_j\Delta t}$, and that after each event it is reset to
	$\ket{\phi_R}=a_1\ket{\phi_1}+a_2\ket{\phi_2}$. The event Kraus operator
	carries $\sqrt{1-\Gamma_j}$ on each generator, so the amplitude for an
	interval of length $n$ is
	$\sqrt{1-\Gamma_1}\,a_1\Gamma_1^{n/2}+\sqrt{1-\Gamma_2}\,a_2\Gamma_2^{n/2}$.
	Choosing
	\begin{equation}
		r\equiv\frac{a_2}{a_1}
		=-\sqrt{\frac{1-\Gamma_1}{1-\Gamma_2}}
		\Big(\frac{\Gamma_1}{\Gamma_2}\Big)^{n^\ast/2}
		\label{eq:forbid}
	\end{equation}
	makes that amplitude vanish at $n=n^\ast$. The process then has
	$d_{n^\ast}=0$ while $d_n>0$ for every $n\ne n^\ast$, so an interevent time
	of exactly $n^\ast$ steps is forbidden although $n^\ast\pm1$ are not.
	
	The ratio in Eq.~\eqref{eq:forbid} is real and negative, a relative phase of
	$\pi$. Normalization in the non-orthogonal generator frame reads
	$a_1^2+a_2^2+2a_1a_2g=1$ with
	$g=\sqrt{(1-\Gamma_1)(1-\Gamma_2)}/(1-\sqrt{\Gamma_1\Gamma_2})$ fixed by
	unitarity, giving $a_1=(1+r^2+2rg)^{-1/2}$. This is not the constraint
	$a_1^2+a_2^2=1$, and the phase of $\pi$ is normalizable precisely because
	the generators are non-orthogonal. At $n^\ast=5$ with
	$(\gamma_1,\gamma_2)=(12,2)$ and $\Delta t=0.05$ this gives $g=0.7017$,
	$r=-0.6238$, $a_1=1.3952$ and $a_2=-0.8704$, and the realized
	$d_{n^\ast}=5.9\times10^{-32}$ against neighbors $4.3\times10^{-3}$ and
	$2.1\times10^{-3}$. The memory is two-dimensional, a single spin one-half.
	
	\begin{proposition}
		\label{prop:forbid}
		Any classical hidden Markov model generating a process with
		$d_{n^\ast}=0$ and $d_n>0$ for all $n\ne n^\ast$ has at least $1+\lceil\sqrt{n^\ast}\,\rceil$ states.
	\end{proposition}
	
	\noindent\emph{Proof.} With $A,b,c\ge0$ entrywise, $d_n=c^{\mathsf T}A^nb$
	vanishes if and only if the support digraph of $A$ has no walk of length $n$
	from $\mathrm{supp}\,b$ to $\mathrm{supp}\,c$. Whether $d_n$ vanishes therefore depends only on the zero
	pattern of $A^n$. That pattern is periodic in $n$ from an index at most
	$(D{-}1)^2+1$ onward, a bound due to Wielandt for primitive
	matrices~\cite{wielandt} that holds for every $D\times D$ nonnegative
	matrix, reducible or not~\cite{jiangshao}. Since $d_n>0$ for every
	$n>n^\ast$, periodicity would force $d_{n^\ast}>0$ if $n^\ast$ reached
	that index. Hence $n^\ast\le(D{-}1)^2$, that is
	$D\ge1+\sqrt{n^\ast}$. $\square$
	
	\medskip
	\noindent By Proposition~\ref{prop:forbid} a spin one-half therefore
	generates processes whose classical cost grows without bound in $n^\ast$, giving $4,5,6,9,11$ states at $n^\ast=5,10,20,50,100$, with $d_{n^\ast}$ vanishing to double precision
	and Hankel rank $2$ throughout. The separation is unbounded, as are those of
	Refs.~\cite{garner,rare,longrange}, but it uses no positive-realization
	theory and its classical side is Perron--Frobenius rather than an entropy
	or a dimension count. Coherence has been maintained over
	several steps of such a simulation experimentally~\cite{interfere}, which
	is the requirement the construction here places on the register.
	The entire resource is one relative phase of $\pi$ and one amplitude
	ratio.
	Figure~\ref{fig:wielandt} shows the mechanism. The bound is tight at $n^\ast=1$ and $2$ but not beyond.
	Exhaustive search over every classical chain with up to four states gives
	the exact minimum as $2$, $3$ and $4$ states at $n^\ast=1$, $2$ and $3$,
	and more than four from $n^\ast=4$ onward, against $2$, $3$, $3$ and $3$
	from the bound, consistent with a true cost linear in $n^\ast$.
	
	\begin{figure}[tb]
		\centering
		\begin{tikzpicture}[
			x=1cm, y=1cm, >={Stealth[length=2mm]},
			st/.style={draw, circle, minimum size=8mm, inner sep=0pt, font=\small,
				fill=blue!12, draw=blue!55, thick},
			lbl/.style={font=\scriptsize}]
			\node[st] (b) at (0,0) {$b$};
			\node[st] (u) at (1.9,0) {$u$};
			\node[st] (c) at (3.8,0) {$c$};
			\draw[->, thick, red!70!black] (b) -- node[above, lbl] {$0\!:\!p_1$} (u);
			\draw[->, thick, red!70!black] (u) -- node[above, lbl] {$0\!:\!p_2$} (c);
			\draw[->, thick, red!70!black] (c) to[bend left=45]
			node[below, lbl] {$0\!:\!p_3$} (u);
			\draw[->, thick, red!70!black] (c) to[out=40, in=-40, looseness=9]
			node[right, lbl] {$0\!:\!p_4$} (c);
			\node[lbl, align=center] at (1.9,-1.35)
			{walk lengths $b\!\to\!c$ are $2,3,4,5,\dots$, one gap, at $n=1$};
		\end{tikzpicture}
		\caption{Why a forbidden interevent time is expensive classically, in the
			edge-emitting convention of Ref.~\cite{zonnios}, in which each transition
			carries an emitted symbol and a weight. Every edge shown belongs to the
			no-event block $A$, an event ending the interval. Since
			$d_n=c^{\mathsf T}A^nb$ with all factors nonnegative, $d_{n^\ast}=0$ holds
			exactly when no walk of length $n^\ast$ joins $\mathrm{supp}\,b$ to
			$\mathrm{supp}\,c$, and because the zero pattern of $A^n$ is periodic from an
			index at most $(D{-}1)^2+1$ onward, retaining walks of every other length
			forces $n^\ast\le(D{-}1)^2$.}
		\label{fig:wielandt}
	\end{figure}
	
	\section{Quadratically many channels}
	\label{sec:advantage}
	The forbidden-time construction uses two amplitudes and one
	phase. With $K$ amplitudes the phases do more than cancel a single term,
	and the memory generates quadratically many decay channels.
	
	\subsection{The memory dimension of a renewal process}
	\begin{proposition}
		\label{prop:hankel}
		Unitarity forces the memory Gram matrix
		$G_{nm}=\langle\varsigma(n)|\varsigma(m)\rangle$ to obey
		$G_{nm}=r_nr_mG_{n+1,m+1}+s_ns_m e^{\mathrm i(\xi_m-\xi_n)}$, with
		$r_n=\sqrt{\Phi(n{+}1)/\Phi(n)}$, $s_n=\sqrt{1-r_n^2}$ and $\xi_n$ the
		phases carried by the reset branch. Setting
		$w_n=\sqrt{d_n}\,e^{\mathrm i\xi_n}$, the recursion telescopes to
		$H_{nm}=\sum_{k\ge0}\overline{w_{n+k}}\,w_{m+k}$ with
		$H=D^{1/2}GD^{1/2}$, $D=\mathrm{diag}\,\Phi$, so that
		$\rank G$ equals the Hankel rank of $w$.
	\end{proposition}
	
	By Proposition~\ref{prop:hankel} the minimal memory dimension is therefore
	the length of the shortest exponential sum whose \emph{modulus squared} is
	$d_n$. The reset phases are the only freedom, and the proposition says
	exactly what they buy.
	
	\subsection{Interference generates channels}
	Generalizing Eq.~\eqref{eq:forbid}, take $K$ terms with real positive bases,
	\begin{equation}
		w_n=\sum_{l=1}^{K}c_l\,\Gamma_l^{\,n/2},\qquad c_l=r_le^{\mathrm i\phi_l},
		\qquad \Gamma_l=e^{-\gamma_l\Delta t}.
		\label{eq:ansatz}
	\end{equation}
	Expanding Eq.~\eqref{eq:ansatz},
	$|w_n|^2=\sum_{l,m}\overline{c_l}c_m(\Gamma_l\Gamma_m)^{n/2}$, and
	since $(\Gamma_l\Gamma_m)^{n/2}=e^{-[(\gamma_l+\gamma_m)/2]n\Delta t}$, this
	is itself a renewal process. The diagonal terms give $K$ channels at rates
	$\gamma_l$ with weights $r_l^2$, and the off-diagonal terms give
	$K(K{-}1)/2$ further channels at the pairwise averages
	$(\gamma_l+\gamma_m)/2$ with weights $2r_lr_m\cos(\phi_l-\phi_m)$. The first
	are populations. The second are interference, and their weights carry the
	sign of the cosine, so they may be negative. A classical mixture of decay
	channels weights each by the probability of decaying through it and has no
	negative weight available, and negativity of this kind has been identified
	as a resource for nonclassical memory advantage~\cite{negativity}.
	Figure~\ref{fig:channels} shows the six channels
	generated by $K=3$ amplitudes, one of which is negative.
	
	Re-expanding $|w_n|^2$ in the predicted channels reproduces it to
	$4.4\times10^{-16}$ for base rates $(12,6,2)$, and the minimal number of
	exponentials fitting $|w|^2$ is $K$ in every case tested, with $K-1$ failing
	by four orders of magnitude or more.
	
	\begin{figure}[tb]
		\centering
		\includegraphics[width=\columnwidth]{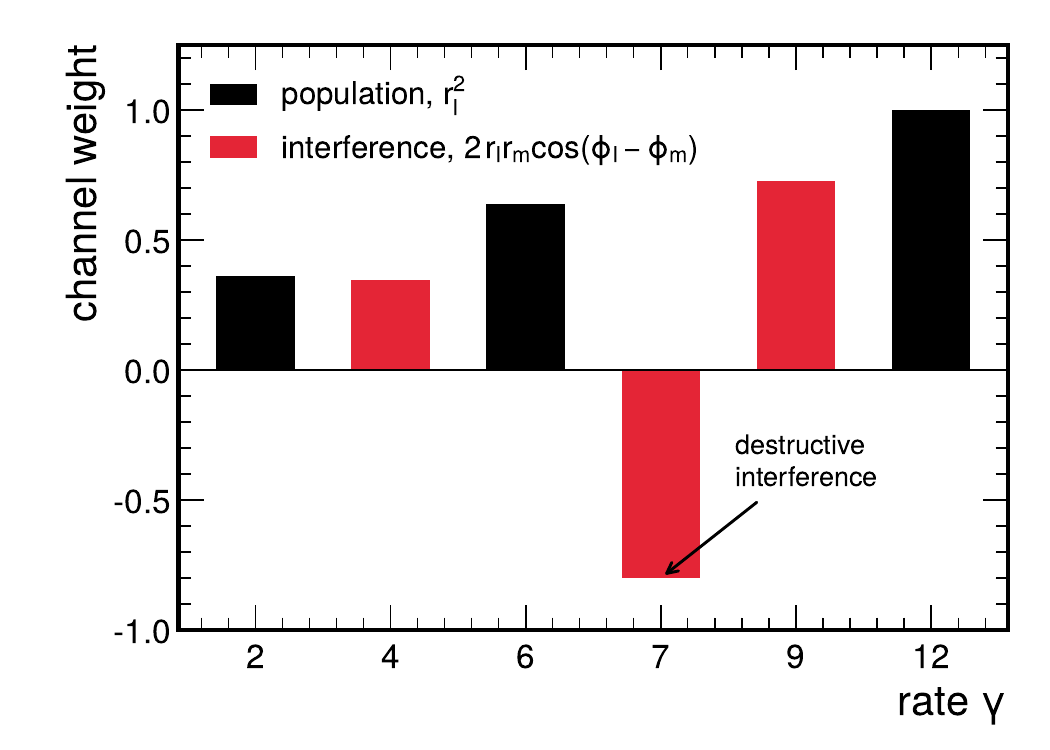}
		\caption{Channels generated by $K=3$ amplitudes with base rates $(12,6,2)$
			and phases $(0,1.1,2.3)$. Bars are equally spaced and the tick labels give
			the true rates. Dark bars at $\gamma=12,6,2$ are populations, weights
			$1.000$, $0.640$, $0.360$. Light bars at the pairwise averages
			$\gamma=9,7,4$ are interference, weights $+0.726$, $-0.800$, $+0.348$. The
			channel at $\gamma=7$, the average of $12$ and $2$, has negative weight
			because $\cos(\phi_1-\phi_3)<0$, so it exists only as destructive
			interference.}
		\label{fig:channels}
	\end{figure}
	
	\subsection{The classical cost}
	\begin{proposition}
		\label{prop:hmm}
		Let $\Phi(n)=\sum_{m=1}^{R}p_m\Gamma_m^{\,n}$ with $R$ distinct $\Gamma_m$
		and all $p_m\ne0$. Then any classical hidden Markov model generating the
		process has at least $R$ hidden states.
	\end{proposition}
	
	\noindent\emph{Proof.} Form $H_{ij}=\Phi(i{+}j)$. On one hand
	$H=V\,\mathrm{diag}(p)\,V^{\mathsf T}$ with $V_{im}=\Gamma_m^{\,i}$
	Vandermonde in distinct nodes, so $\rank H=R$. On the other hand a classical
	model with $K'$ states and no-event transition block $T$ gives
	$\Phi(n)=\pi^{\mathsf T}T^{n}\mathbf 1$, so
	$H_{ij}=(\pi^{\mathsf T}T^{i})(T^{j}\mathbf 1)$ factors through
	$\mathbb R^{K'}$ and $\rank H\le K'$. Hence $K'\ge R$. $\square$
	
	\medskip
	\noindent This is the Hankel-rank characterization of minimal
	realizations~\cite{jaeger,balle}, included to fix notation. Its quantum
	counterpart is the completely positive realization
	problem~\cite{monras}, which is not always solvable in finite
	dimension~\cite{fanizza}. The rates
	$(\gamma_l+\gamma_m)/2$ are distinct for generic base rates, a collision
	requiring $\gamma_l+\gamma_m=\gamma_p+\gamma_q$, which is a measure-zero
	condition. Combining with the previous subsection, a memory of dimension $K$
	generates processes that any classical hidden Markov model needs at least
	$K(K{+}1)/2$ states to reproduce, that is $3$ against $2$ at $K=2$ and $105$
	against $14$ at $K=14$. Both sides of this comparison are finite hidden
	Markov models, so the statement involves no convention about what counts as
	a classical memory state.
	
	The lower bound is attained. When every relative phase satisfies
	$|\phi_l-\phi_m|<\pi/2$ all channel weights are positive, and the process
	is then generated by an explicit classical machine with exactly $R$
	states.
	
	\begin{proposition}
		\label{prop:exact}
		Let $\Phi(n)=\sum_{m=1}^{R}p_m\Gamma_m^{\,n}$ with $R$ distinct $\Gamma_m$
		and all $p_m>0$. Then the minimal classical hidden Markov model generating
		the process has exactly $R$ hidden states.
	\end{proposition}
	
	\noindent\emph{Proof.} Proposition~\ref{prop:hmm} gives at least $R$. For
	the converse, let the machine enter state $m$ after each event with
	probability $p_m$, and in state $m$ emit $0$ and remain with probability
	$\Gamma_m$ or emit an event with probability $1-\Gamma_m$. The run length
	in state $m$ is geometric with survival $\Gamma_m^{\,n}$, so the machine
	reproduces $\Phi(n)=\sum_mp_m\Gamma_m^{\,n}$ exactly with $R$ states. The
	entry probabilities are a valid distribution because $p_m>0$ and
	$\Phi(0)=1$. $\square$
	
	\medskip
	\noindent A memory of dimension $K$ therefore generates processes whose
	classical cost is not merely bounded below by $K(K{+}1)/2$ but equal to
	it, against a quantum cost of at most $K$. The separation is attained
	rather than being a comparison of lower bounds. We verified both sides
	numerically for $K=2$ to $5$, the Hankel rank equaling $R$ at $50$-digit
	precision and the explicit $R$-state machine reproducing $d_n$ to
	$6\times10^{-17}$. When some phase difference exceeds $\pi/2$ a channel
	weight turns negative, the construction above is no longer a valid
	machine, and the classical cost is then bounded below by $R$ but may be
	larger, since a nonnegative realization can require more states than the
	Hankel rank. Verified at $60$-digit precision, $\rank H$ equals
	the distinct-rate count exactly.
	
	The resource is the set of relative phases $\phi_l$, which are what make
	$\cos(\phi_l-\phi_m)$ available and hence generate the off-diagonal channels
	and their negative weights. A classical memory carries populations but no
	relative phases.
	
	\section{Realization}
	\label{sec:realization}
	The constructions above are implementable in the collective-spin register of Fig.~\ref{fig:circuit}. Driving the memory by the total-spin generators alone would close the
	dynamical Lie algebra on $\mathfrak{su}(2)$, of dimension three
	independently of $\nm$. The gate set of Fig.~\ref{fig:circuit} also
	rotates and couples the work register, and its dynamical Lie algebra is
	the full $\mathfrak{su}(D)$ with
	$D=2^{\nw}(\nm{+}1)$ at every size computed, namely $\dim=15,35,63,99$ for
	$\nw=1$ and $\nm=1,\dots,4$, and $63,143,255,399$ for $\nw=2$, in each case
	$D^2-1$, so every unitary on the joint register is reachable. Reachability
	is not by itself a circuit, so we synthesize one, in the sense of
	Refs.~\cite{binder,liu}, where unitary implementations of quantum models
	of stochastic processes are constructed. Fitting a depth-$L$ block to the two target Kraus
	operators of the $n^\ast=5$ forbidden-time machine, depth
	$L=3$ with twenty parameters reproduces them to $5.6\times10^{-16}$ in
	every block, while $L=2$ with fourteen fails at $7.1\times10^{-2}$, so
	the depth requirement is real rather than an artifact of the optimizer.
	The failure is visible in the process itself. At $L=3$ the realized
	interevent distribution matches the target to $1.4\times10^{-15}$ with
	$d_{n^\ast}=5.9\times10^{-32}$, whereas at $L=2$ the forbidden time is
	not forbidden, $d_{n^\ast}$ being $2.4\times10^{-4}$.
	
	The admissible phase assignments are characterized exactly. Writing
	$c_l=r_le^{\mathrm i\phi_l}$, the normalized weight matrix
	$C_{lm}=a_{lm}/(2\sqrt{a_{ll}a_{mm}})$, $C_{ll}=1$, equals
	$\cos(\phi_l-\phi_m)$ and is therefore the Gram matrix of $K$ unit vectors
	in the plane, of rank at most two. Allowing the emitted event symbol to
	correspond to $J$ work-register outcomes rather than one relaxes this to
	rank at most $2J$, and $J=\lceil K/2\rceil$ removes the constraint
	altogether, so the covered family is full measure in weight space at every
	$K$. Since $J\le2^{\nw}-1$, two work qubits suffice through $K=6$ and three
	through $K=14$, which is $105$ channels. The construction is explicit, an
	eigendecomposition of $C$ with no optimization, and the threshold is exact
	through $K=14$.
	
	\section{Relation to previous work}
	\label{sec:prior}
	That non-orthogonal memory states reduce the memory a quantum model needs is
	due to Gu \emph{et al.}~\cite{gu}, building on the computational-mechanics
	account of causal states~\cite{cy,sc,marzen}. Elliott and Gu~\cite{elliottgu}
	treat renewal processes in continuous time with an infinite-dimensional
	register of finite entropy, and Elliott \emph{et al.}~\cite{elliott} give
	the dual Poisson machine at $\chi=2$, which is the case $K=2$ of Eq.~\eqref{eq:ansatz}. We take that construction as the starting point
	and extend it to arbitrary $K$. Neither reference counts the channels a
	memory generates, and neither builds overlaps from a finite exponential sum.
	
	Zonnios, Boyd and Binder~\cite{zonnios} derive lower bounds on both costs
	from a spectral invariant $|\Lambda|$ of the process,
	$c_C\ge\log\lceil|\Lambda|^{1/2}\rceil$ against
	$c_Q\ge\log\lceil|\Lambda|^{1/4}\rceil$, so that the classical bound is
	quadratically larger than the quantum one, and attribute the gap to
	coherence. Their invariant counts eigenvalues of a second-moment transfer
	operator, whereas the quantity of Proposition~\ref{prop:hmm} is a
	first-moment Hankel rank, and on
	the family here the two can be compared directly. Counting every distinct
	eigenvalue of the transfer operator, which can only overstate their bound,
	the two coincide at $K=2$ and $K=3$, and from $K=4$ the Hankel bound is
	strictly tighter, giving $10,15,21$ at $K=4,5,6$ against at most
	$9,12,16$. Because both of
	their statements are lower bounds, the quadratic relation holds between
	the bounds rather than between the costs, and their worked example attains
	$c_Q=\log2$ against $c_C=\log3$, a separation but not a quadratic one.
	Proposition~\ref{prop:exact} exhibits, among processes with positive
	weights, a family at every $K$ whose costs themselves separate
	quadratically, with the classical cost equal to
	$K(K{+}1)/2$ exactly and the quantum cost at most $K$. Appendix~E of
	that reference leaves open how to choose the phase degrees of freedom
	optimally, which the reset-branch criterion above answers for renewal processes.
	We are not aware of an antecedent for the forbidden-time construction, and the index of convergence of a nonnegative matrix appears not
	to have been used previously as a classical memory lower bound.
	
	None of this prior work concerns a physical memory architecture. A spectral
	bound establishes that a small model exists, but does not produce a circuit
	that realizes it.
	
	\section{Discussion}
	\label{sec:discussion}
	The separations here are statements about model size rather than runtime. A
	collective memory is efficiently classically simulable, since its dimension
	is $\nm+1$, and that is deliberate. Simulation cost and memory dimension are
	independent resources, so a model may be efficiently simulable and still
	require unboundedly less memory than any classical model of the same
	process. The quadratic separation is a theorem rather than a claim about
	classical hardness, and the forbidden-time separation is unbounded with a
	classical side that is Perron--Frobenius.
	
	Two limitations are load bearing. The collective ansatz restricts the
	achievable Kraus operators, and although the renewal subfamily is
	characterized exactly, the non-renewal case remains open. On
	fixed-connectivity hardware the coupling decomposes as
	$\prod_q R_{ZZ}(t/j)$, a star from each work qubit to every memory qubit, so
	one rollout costs $L\,B\,\nw\,\nm$ two-qubit gates. The collective
	construction buys a memory whose dimension is a resource, not a cheaper
	circuit, and on a platform with a collective bus the same coupling is a
	single physical operation.
	
	\section*{Code availability}
	Code that regenerates every number in this Letter, using \texttt{numpy},
	\texttt{scipy} and \texttt{mpmath}, and checks each against the value
	printed here, is available from the author and will be deposited in a
	public repository with a DOI upon publication.
	
	\section*{Acknowledgments}
	
	The author thanks Dirk Kr\"ucker for comments and discussions, and acknowledges support from the Helmholtz Association HGF (Germany), the Hamburgische	Investitions- und F\"orderbank (IFB) (Germany), and the Ministry of Science, Research and Culture of the State of Brandenburg within the Center for Quantum Technologies and Applications (CQTA) (Germany).

\end{document}